\documentclass[conference,compsocconf]{IEEEtran}

\ifCLASSINFOpdf

\else

\fi
\usepackage{fixltx2e}
\usepackage{macros}
\usepackage{amsfonts}
\usepackage{amsmath}
\usepackage{amstext}
\usepackage{amssymb}
\usepackage{latexsym}
\usepackage[ruled]{algorithm}
\usepackage{algpseudocode}
\usepackage{mathtools}
\usepackage{todonotes}
\usepackage{txfonts}
\usepackage{booktabs} 
\usepackage{bigstrut} 
\usepackage{graphicx}
\usepackage{multicol}
\usepackage{multirow}
\usepackage{relsize}
\usepackage{xcolor,colortbl}
\usepackage{vwcol}  
\graphicspath{{figures/}}

\usepackage{color}

\usepackage{tikz}
\usetikzlibrary{calc,matrix,patterns,shadings,backgrounds,fit}

\usepackage{ifthen}

\usepackage{stmaryrd}
\usepackage{macros}

\newboolean{HIGHCOMM}
\setboolean{HIGHCOMM}{false}
\newcommand \redhighlight[1]{\ifthenelse{\boolean{HIGHCOMM}}{\textcolor{red}{#1}}{#1}}
\newboolean{HIGHLIGHT}
\setboolean{HIGHLIGHT}{false}
\newcommand \yhl[1]{\ifthenelse{\boolean{HIGHLIGHT}}{\hl{#1}}{#1}}

\newtheorem{lemma}{Lemma}

\newtheorem{example}{Example}

\newtheorem{problem}{Problem}
\newtheorem{definition}{Definition}

\newtheorem{theorem}{Theorem}

\begin{document}

\title{An Efficient Algorithm for Monitoring Practical TPTL Specifications}
\author{
	\IEEEauthorblockN{Adel Dokhanchi, Bardh Hoxha, Cumhur Erkan Tuncali, and Georgios Fainekos}
	Arizona State University, Tempe, AZ, U.S.A.\\
	Email: \{adokhanc,bhoxha,etuncali,fainekos\}@asu.edu
}

\maketitle             
\thispagestyle{empty}
\pagestyle{empty}

\begin{abstract} 
We provide a dynamic programming algorithm for the monitoring of a fragment of Timed Propositional Temporal Logic (TPTL) specifications.
This fragment of TPTL, which is more expressive than Metric Temporal Logic, is characterized by independent time variables which enable the elicitation of complex real-time requirements. 
For this fragment, we provide an efficient polynomial time algorithm for off-line monitoring of finite traces. 
Finally, we provide experimental results on a prototype implementation of our tool in order to demonstrate the feasibility of using our tool in practical applications.


\end{abstract}

\IEEEpeerreviewmaketitle
\section{Introduction}
In Cyber-Physical Systems (CPS), many safety critical components of the system are controlled by embedded computers which interact with the physical environment.
Due to the safety-critical nature of these applications, it is important to verify their correctness during system development stages.
However, the verification problem for CPS with respect to safety requirements is undecidable, in general \cite{hs_gf:Alur95algorithmic}.
An alternative to formal verification is semi-formal model-based testing and monitoring of CPS. 
We utilize formal logic, in order to formally specify real-time requirements.

Metric Temporal Logic (MTL) was introduced to provide the formalization of real-time specifications \cite{Koymans1990}.
Since its introduction, MTL and its variants have been used in the verification of real-time systems \cite{OuaknineW08}.
Several tools, such as \staliro \cite{Annapureddy2011} and Breach \cite{Donze2010}, have been developed by the academic community for the purpose of semi-formal verification of MTL specifications. 
These tools use off-line and on-line monitoring algorithms to check whether the execution trace of a CPS satisfies/falsifies an MTL formula.
In off-line monitoring, the execution trace 
is finite and generated by running the system for a bounded amount of time. 
Then, the off-line monitor checks whether the execution trace satisfies the specification.
On the other hand, an on-line monitor runs simultaneously with the system.
In this paper, we consider off-line monitoring of TPTL specifications.

The time complexity of off-line monitoring for MTL is linear to the size of a finite system trace and linear to the size of MTL formula. 
Several algorithms using dynamic programming \cite{Fainekos2012} or sliding windows \cite{DonzeAF2013stl} have been proposed for MTL monitoring of CPS.
In this paper, we consider TPTL specifications which are more expressive than MTL specifications \cite{BouyerCM10}.
TPTL is an extension of Linear Temporal Logic (LTL) with freeze quantifiers represented as ``$x.$''.
A freeze quantifier $x.$ assigns to time variable $x$ the ``current'' time stamp when the corresponding subformula $x.\varphi(x)$ is evaluated \cite{AlurH94}.
Then, the time value (stored in $x$) can be evaluated inside time constraints which are linear inequalities over the time variables.
 
Since its introduction, two semantics where considered for TPTL \cite{AlurH94,BouyerCM10}. 
Alur's semantics \cite{AlurH94} allows two time variables in time constraints (for example $x+1\le y+4$).
In contrast, Raskin's semantics allows only one time variable in the time constraint ($x\le4$) and implicitly considers the current time as the second time variable \cite{BouyerCM10,Raskin99}.
Since the latter semantics was first considered by Jean-François Raskin in \cite{Raskin99}, we will refer to it as ``Raskin's TPTL semantics'' in this paper.
Raskin's TPTL semantics was mentioned with alternative terms such as ``Timed LTL" in \cite{KristoffersenPA03}.
In another line of work, in \cite{DeshmukhMP15}, the authors augmented Alur's time constraints with more complex temporal-special predicates to define the closeness property of two different CPS trajectories.
However, the authors in \cite{DeshmukhMP15} did not provide a TPTL monitoring algorithm.

Since TPTL subsumes MTL, it is expected that the monitoring problem of TPTL is computationally more complex \cite{FengLQ15}.
It has been proven that monitoring of a finite trace with respect to Alur's TPTL specification is PSPACE-hard \cite{Markey2006MCR}. 
In \cite{Markey2006MCR}, the authors transform a Quantified Boolean Formula (QBF), which is PSPACE-hard, into a TPTL formula with real value time variables.
A similar complexity result (PSPACE-hard) for Raskin's TPTL semantics is obtained for integer time variables in \cite{FengLQ15}.
It is mentioned in \cite{FengLQ15} that in order to obtain a polynomial time algorithm for TPTL monitoring (path checking), we need to fix the number of time variables.
In other words, if the number of time variables is bounded then the finite trace monitoring will be polynomial to the size of the TPTL formula.
However, in \cite{FengLQ15}, the authors did not provide any applicable algorithm for TPTL monitoring and they focused only on the complexity class.

In this work, we move one step further from \cite{FengLQ15}, and allow the number of time variables to be arbitrary, but they must be independent to each other\footnote{In Section \ref{fragments}, Definition \ref{independent}, we introduce independent time variables.}.
For this fragment of TPTL, we provide an efficient TPTL monitoring algorithm which has time complexity quadratic in the length of the finite trace.
In addition, the runtime of the algorithm is proportional to the number of time variables in TPTL.

In terms of related work, a rewriting based algorithm for TPTL has been provided in \cite{ChaiS13}.
In \cite{ChaiS13}, the authors did not evaluate the time complexity of their proposed algorithm.
The rewriting technique was used for on-line monitoring of TPTL specifications in \cite{HakanssonJL03}.
The authors used the relativization of TPTL formula with respect to the sequence of observed states \cite{HakanssonJL03}, and it was reported that the time complexity is exponential to the size of TPTL formula \cite{HakanssonJL03}. 
To the best of our knowledge, our paper is the first work where an efficient and practical TPTL off-line monitoring algorithm is provided.

\section{Preliminaries}
\redhighlight{We assume a sampled representation of system behavior with a discrete trace as the input to the monitoring algorithm. We utilize the notion of Timed State Sequences (TSS) \cite{AlurH94} to represent the sampled behavior of a system using a digital clock.} 
We interpret TPTL formulas over TSS. Assume $AP=\{a,b,\cdots\}$ is a set of atomic propositions, $\Re_+$ is the set of non-negative real numbers, and $\Ne$ denotes non-negative integers.
\begin{definition}[State and Time Sequences \cite{AlurH94}]
	\label{def:t&s}
	A state sequence $\dsig=\dsig_0\dsig_1\dsig_2\cdots$ is an infinite sequence of states $\dsig_i\subseteq AP$, where $i\in \Ne$. A (sampled) time sequence $\tau=\tau_0\tau_1\tau_2\dots$ is an infinite sequence of time stamps $\tau_i\in\Re_+$, where $i\in \Ne$.
\end{definition} 
 We assume that the time sequence $\tau$ is:
 \begin{enumerate}
 \item {\bf Initialized}, which means that the start up time is zero ($\tau_0=0$).
 \item {\bf Monotonic}, which means that $\tau_i\le \tau_{i+1}$ for all $i\in\Ne$.
 \item {\bf Progressive}, which means that for all $t\in\Re_+$ there is some $i\in\Ne$ such that $\tau_i>t$.
\end{enumerate}
\begin{definition}[Timed State Sequence (TSS) \cite{AlurH94}]
	\label{def:tss}
	A timed state sequence $\rho=(\dsig,\tau)$ is a pair consisting of a state sequence $\dsig$ and a time sequence $\tau$ where $\rho_0\rho_1\rho_2\cdots=(\dsig_0,\tau_0)(\dsig_1,\tau_1)(\dsig_2,\tau_2)\cdots$.
\end{definition} 
Given an infinite TSS $\rho$, we consider a finite prefix of $\rho$ as a finite TSS. 
The symbol $\hat{\rho}=(\hat{\dsig},\hat{\tau})$ is used to denote a finite TSS with the size of $|\hat{\rho}|=|\hat{\dsig}|=|\hat{\tau}|$. 
In this paper, we consider the monitoring of finite TSS with the size of $|\hat{\rho}|$ which is equal to the number of simulation/execution samples.

\subsection{TPTL Syntax and Semantics}
To prevent any confusion in the presentation, we consider Raskin's TPTL semantics \cite{Raskin99,BouyerCM10}\footnote{We will explain in Section \ref{fragments} why we chose Raskin's semantics.}.
TPTL is an extension of LTL that enables the formalization of real-time properties by including time variables and a freeze time quantifier \cite{AlurH94}.

\begin{definition}[Syntax for $TPTL$]
	\label{def:tptlsyn}
 The set of  TPTL  formulas $\varphi$ over a finite set of atomic propositions ($AP$) and a finite set of time variables ($V$) is inductively defined according to the following grammar:
\end{definition} 
	$$\varphi\;::=\;\top \; | \; a \; | \;  x\sim r \; | \; \neg\varphi \; | \; \varphi_1 \wedge \varphi_2\; | \; \varphi_1 \vee \varphi_2 \; | \; \bigcirc\varphi \; | \; \varphi_1 U \varphi_2  \; | \;  x.\varphi$$
where $x\in V$, $r\in \Re_+$, $a\in AP$, and $\sim$ $\in\{\le,<,=,>,\ge\}$, and $\top$ is the symbol for ``True''. \\
The time constraints of TPTL are represented in the form of $x\sim r$.
The freeze quantifier $x.$ assigns the current time of the formula's evaluation (at each sampled time $\tau_i$) to the time variable $x$. 
A TPTL formula is $closed$ if every occurrence of a time variable is within the scope of a freeze quantifier \cite{AlurH94}. 
In TPTL specifications, we always deal with closed formulas.

We note that ``False'' is represented as $\bot\equiv\neg\top$ and ``Implication'' is represented as $\varphi_1\rightarrow\varphi_2\equiv\neg\varphi_1\vee\varphi_2$. 
For all formulas $\psi$, $\phi$, $\Diamond\psi\equiv\top U\psi$ (Eventually $\psi$), $\Box\psi\equiv\neg\Diamond\neg\psi$ (Always $\psi$), and $\psi R\phi\equiv\neg(\neg\psi U\neg\phi)$ ($\psi$ Releases $\phi$) are defined in the conventional way. Since we focus on off-line monitoring, we only consider the TPTL semantics for finite traces.
\begin{definition}[Discrete-Time Semantics for $TPTL$]
	\label{def:sem4tptl}
	Let  $\hat{\rho}=(\hat{\dsig},\hat{\tau})$ be a finite TSS and $i\in\Ne$ where $i<|\hat{\rho}|$ is the index of the current sample, $a\in AP$, $\varphi\in TPTL$, and an environment $\varepsilon:V\rightarrow\Re_+$. 
	The satisfaction relation $(\hat{\rho},i,\varepsilon)\models\varphi$ is defined recursively as follows:
\end{definition}
\begin{enumerate}
	\item []$(\hat{\rho},i,\varepsilon)\models\top$
	\item []$(\hat{\rho},i,\varepsilon)\models a$ iff  $a\in \dsig_i$
	\item []$(\hat{\rho},i,\varepsilon)\models\neg\varphi$  iff  $(\hat{\rho},i,\varepsilon)\not\models\varphi$
	\item []$(\hat{\rho},i,\varepsilon)\models\varphi_1\wedge\varphi_2$  iff  $(\hat{\rho},i,\varepsilon)\models\varphi_1$ and $(\hat{\rho},i,\varepsilon)\models\varphi_2$	
	\item []$(\hat{\rho},i,\varepsilon)\models\varphi_1\vee\varphi_2$  iff  $(\hat{\rho},i,\varepsilon)\models\varphi_1$ or $(\hat{\rho},i,\varepsilon)\models\varphi_2$
	\item []$(\hat{\rho},i,\varepsilon)\models\bigcirc\varphi$  iff $(\hat{\rho},i+1,\varepsilon)\models\varphi$ and $i<(|\hat{\rho}|-1)$
	\item []$(\hat{\rho},i,\varepsilon)\models\varphi_1U\varphi_2$  iff $\exists j, i\le j<|\hat{\rho}|$ s.t. $(\hat{\rho},j,\varepsilon)\models\varphi_2$ and $\forall k, i\le k<j$ it holds that $(\hat{\rho},k,\varepsilon)\models\varphi_1$
	\item []$(\hat{\rho},i,\varepsilon)\models x\sim r$ iff $(\tau_i-\varepsilon(x))\sim r$ i.e. \\(current time stamp)  $-~\varepsilon(x)\sim r$
	\item []$(\hat{\rho},i,\varepsilon)\models x.\varphi$ iff $(\hat{\rho},i,\varepsilon[x:=\tau_i])\models\varphi$
\end{enumerate}

The semantics of TPTL are defined over an evaluation function $\varepsilon:V\rightarrow\Re_+$ which is an environment for the time variables. Assume $x=r$ where $x\in V$, and $r\in\Re_+$, then we have $\varepsilon(x)=r$. 
Given a variable $x\in V$ and a real number $q\in \Re_+$, we denote the environment with $\varepsilon'=\varepsilon[x:=q]$  which is equivalent to the environment $\varepsilon$ on all time variables in $V$ except variable $x$.
The assignment operation $x:=q$ changes the environment $\varepsilon$ to the new environment $\varepsilon'$.
Formally, $\varepsilon'(y)=\varepsilon(y)$ for all $y\ne x$ and $\varepsilon'(x)=q$. 
We write {\bf 0} for the ({\bf zero}) environment such that {\bf 0}$(x)=0$ for all $x\in V$.
We say that $\hat{\rho}$ satisfies $\varphi$ ($\hat{\rho}\models\varphi$) iff $(\hat{\rho},0,${\bf 0}$)\models\varphi$.
A variable ``$x$'' that is bounded by a corresponding {\it freeze quantifier} ``$x.$'' saves the local temporal context $\tau_i$ (now) in ``$x$''. 
Assume $\varphi(x)$ is a formula with a free variable $x$. 
The TSS $\hat{\rho}$ satisfies $x.\varphi(x)$ if it satisfies $\varphi(\tau_0=0)$, where $\varphi(0)$ is obtained from $\varphi(x)$ by replacing all the free occurrences of the variable $x$ with constant 0 \cite{AlurH94}.


\subsection{TPTL Fragments}
\label{fragments}
In this section, we introduce a TPTL fragment for which we have developed a monitoring algorithm.
This restriction is crucial for obtaining the polynomial runtime of the algorithm.
\begin{definition}[Independent Time Variable]
\label{independent}
A time variable $x$ is independent if it is 
in the scope of 
only one freeze quantifier $x.$ and no other time variable is 
in the scope of 
the corresponding freeze quantifier ($x.$).
\end{definition}

For example in $x.(\psi(x)\vee\Diamond y.\varphi(x,y))$, neither $x$ nor $y$ is independent. 
This is because $x$ is within the scope of the freeze time quantifiers $x.$ in $x.(\psi(x)\vee\Diamond y.\varphi(x,y))$ and $y.$ in $y.\varphi(x,y)$.
Similarly, $y$ is not the only time variable that is within the scope of $y.$ in $y.\varphi(x,y)$. 
However, both $x$ and $y$ are independent in $x.(\psi(x)\vee\Diamond y.\varphi(y))$.

Now we explain why we focus on Raskin's semantics in our monitoring algorithm.
In Raskin's semantics, each time constraint contains a single time variable (see Definition \ref{def:tptlsyn}).
However, in Alur's semantics each time constraint contains two time variables \cite{AlurH94}.
In Alur's semantics, time variables in the same constraint are dependent to each other.
As a result, in order to benefit from independent time variables, we should consider Raskin's semantics.
\begin{definition}[Encapsulated TPTL formula]
	 Encapsulated TPTL formulas are TPTL formulas where all the time variables are independent.
\end{definition}

\redhighlight{In other words, an encapsulated formula is a closed formula in which every sub-formula has at most one free time variable.}
\begin{definition}[Frozen Subformula]
\label{FRZ}
Given an encapsulated TPTL formula~$\Phi$, a frozen subformula $\phi$ of $\Phi$ is a subformula which is bounded by a freeze quantifier corresponding to (an independent) time variable.
\end{definition}

\redhighlight{In encapsulated formulas, all the closed subformulas are frozen. 
For example the formula $x.(\psi(x)\vee\Diamond y.\varphi(x,y))$ is not an ``encapsulated'' formula because $y.\varphi(x,y)$ is not frozen since $x,y$ are not independent.}
Here are two TPTL formulas $\varphi_1$,$\varphi_2$ that look similar but only one of them is encapsulated.
\begin{itemize}
\item $\varphi_1=\Box x.\Diamond(a\wedge x\le10\wedge y.\Box(\mathbf{y\le2}\wedge y\ge1\wedge b))$
\item $\varphi_2=\Box x.\Diamond(a\wedge x\le10\wedge y.\Box(\mathbf{x\le2}\wedge y\ge1\wedge b))$
\end{itemize}
In the above, $\varphi_1$ is encapsulated, but $\varphi_2$ is not encapsulated since $y.\Box(x\le2\wedge y\ge1\wedge b)$ where $x\le2$ is inside the scope of ``$y.$''.
\begin{lemma}
	Any MTL formula can be represented by an ``encapsulated'' TPTL formula.
\end{lemma}
\begin{proof}
	Each time interval of an MTL temporal operator can be represented with a unique time variable which is independent of the rest of time variables. The syntactic modification works as follows: every MTL formula of the form $\varphi=\psi U_{[l,u]}\phi$ can be recursively represented as the following TPTL formula $\varphi= x.(\psi U(x\ge l\wedge x\le u \wedge \phi))$.
	The resulting TPTL formula is encapsulated. 
\end{proof}
\begin{lemma}
	MTL is less expressive than ``encapsulated'' TPTL formulas.
\end{lemma}
\begin{proof}
It is proven in \cite{BouyerCM10} that the following TPTL formula, which is evidently encapsulated, cannot be expressed by any MTL formula \cite{BouyerCM10}: $\psi=x.\Diamond(a\wedge x\le1\wedge\Box(x\le1\rightarrow\neg b))$
\end{proof}
In the rest of the paper, we focus on the following problem:
\begin{problem}
	Given a finite TSS $\hat{\rho}$ and an ``encapsulated'' TPTL formula $\varphi$, check whether $\hat{\rho}$ satisfies $\varphi$ ($\hat{\rho}\models\varphi$).
\end{problem}
\section{Monitoring Encapsulated TPTL Formulas}
\subsection{TPTL Representation}
\label{datastructure}
In the following, we will describe the data structure that will be utilized to capture the solution for the TPTL monitoring problem.
We store each TPTL formula in a binary tree data structure. 
Consider the following example:
\begin{example}
	\label{mtvExmp}
	Assume $AP=\{a,b\}$ and let\\
	$\phi=\Box x.\Diamond( (x\le1\rightarrow a) \wedge y.\Diamond(y\le1\rightarrow\neg b))$
	\\$\phi\equiv \Box x.\Diamond((x\le1\rightarrow a)\wedge y.\psi_1(y))\equiv \Box x.\psi_2(x)$ \\
	where we use $\psi_1$ and $\psi_2$ to simplify the presentation:\\ $\psi_1(y)\equiv\Diamond(y\le1\rightarrow\neg b)$ \\ $\psi_2(x)\equiv\Diamond((x\le1\rightarrow a)\wedge y.\psi_1(y))$
\end{example}
In this example, we have two independent time variables $x$ and $y$. 
The binary tree of Example \ref{mtvExmp} is depicted in Fig. \ref{fig:Tree}.
There, the thirteen nodes correspond to thirteen subformulas.
\begin{figure}[!t]
	\vspace{-5pt} 
	\vspace{-10pt} 
	\centering 
	\includegraphics[width=7cm]{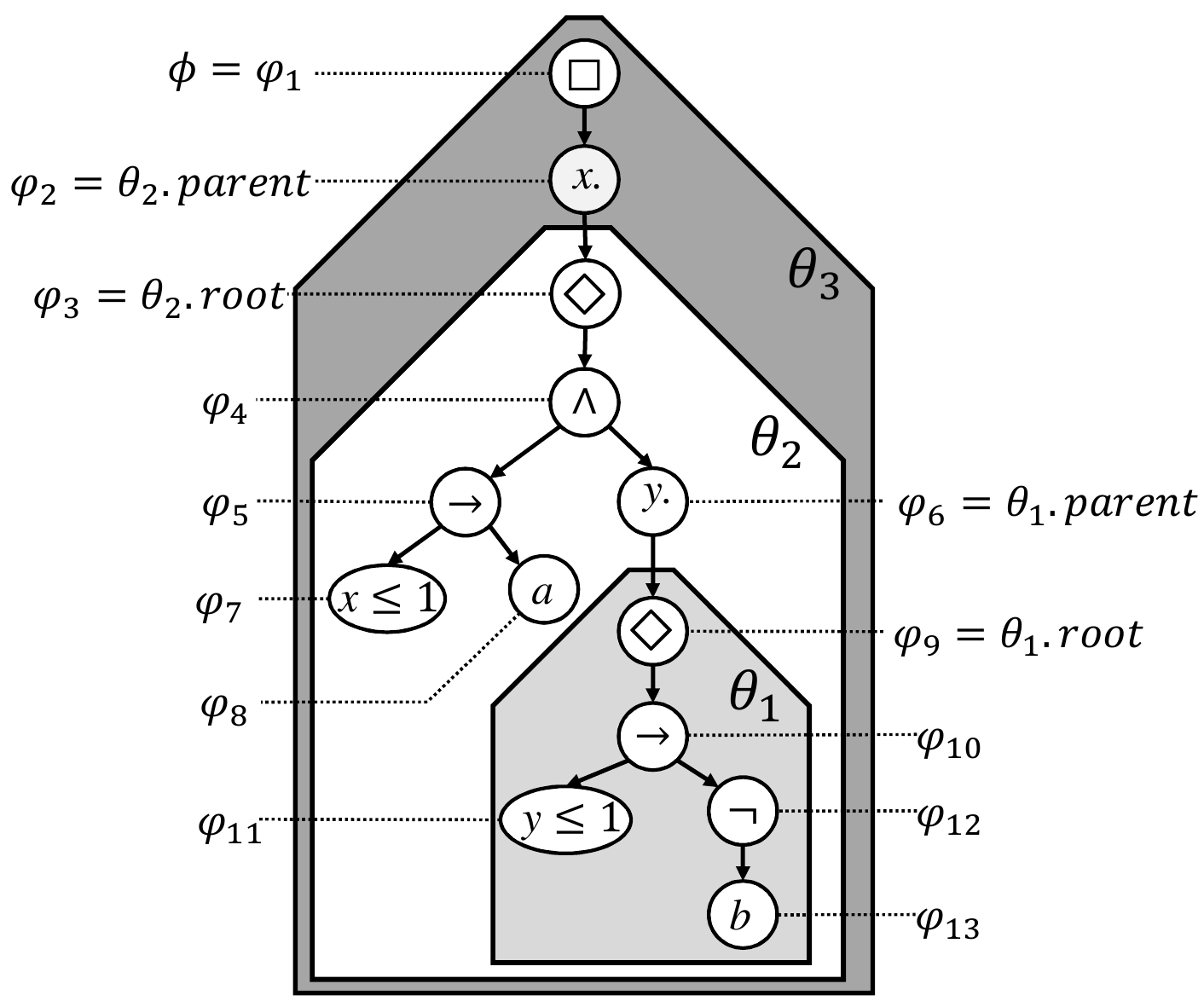}
	\vspace{-10pt}
	\caption{Binary tree of Example \ref{mtvExmp} ($\phi$) with three subtrees corresponding to sets of subformulas $\theta_1,\theta_2,\theta_3$.}
	\label{fig:Tree}
	\vspace{-15pt} 
\end{figure}

In Fig. \ref{fig:Tree}, each subformula $\varphi_i$ has a node corresponding to the highest operator for $\varphi_i$.
In addition, for each subformula $\varphi_i$ we assign an index $i$.
The order of indexes is generated according to a topological sort where parents have lower index values than children.
Therefore, the original subformula $\phi$ obtains the index 1 because it is the first visited.
To evaluate each node's $\top/\bot$ value we need to evaluate its children's $\top/\bot$ value before, this is because of the TPTL recursive semantics (see Definition \ref{def:sem4tptl}).
If we evaluate the nodes in the decreasing order of indexes, we would be able to evaluate all the children before their parents.

Now, we must partition the formula tree into subtrees rooted by the freeze time operators.
Since in Example \ref{mtvExmp}, we have two independent time variables, we created 2+1 subtrees (two for time variables and one for the original formula).
Each subtree contains a set of subformulas. 
These subformulas and their corresponding subtrees $\theta_1,\theta_2,\theta_3$ are shown in Fig. \ref{fig:Tree} with different colors: 

The set $\theta_1$ contains subformulas rooted at node $\varphi_9$ represented in the {\bf light-gray} subtree. The set $\theta_1$ contains the subformulas of $y.\psi_1(y)$ as follows $\theta_1=\{\Diamond(y\le1\rightarrow\neg b),y\le1\rightarrow\neg b,y\le1,\neg b,b\}=\{\varphi_9,\varphi_{10},\varphi_{11},\varphi_{12},\varphi_{13}\}$.

The set $\theta_2$ contains subformulas rooted at node $\varphi_3$ represented in the {\bf white} subtree. The set $\theta_2$ contains the subformulas of $x.\psi_2(x)$ as follows $\theta_2=\{\Diamond((x\le1\rightarrow a)\wedge y.\psi_1(y)),(x\le1\rightarrow a)\wedge y.\psi_1(y),(x\le1\rightarrow a),y.\psi_1(y),x\le1,a\}=\{\varphi_3,\varphi_4,\varphi_5,\varphi_6,\varphi_7,\varphi_8\}$.

The set $\theta_3$ contains subformulas rooted at node $\varphi_1$ represented in {\bf dark-gray} subtree. The set $\theta_3$ contains the subformulas of $\theta_3=\{\Box x.\psi_2(x)\mbox{ , }x.\psi_2(x)\}=\{\varphi_1,\varphi_2\}$.

Each of the subtrees $\theta_1$ and $\theta_2$ have distinguished fields referencing to (the index of) $parent$ and $root$ nodes which are represented in Fig. \ref{fig:Tree} as follows:\\
1) $\theta_1.parent=6$ and $\theta_1.root=9$.\\
2) $\theta_2.parent=2$ and $\theta_2.root=3$.

Note that $\theta_1$ is subformula of $\theta_2$, and $\theta_2$ is subformula of $\theta_3$. This ordering is very important for our algorithm. We created these subtrees because each frozen subformula can be separately evaluated.
Therefore, we can guarantee the polynomial runtime.
The method will be described in details in Section \ref{MA}.
\subsection{Monitoring Table}
\label{MT}
We assume that the sampled system output is mapped (projected) on a {\it finite} TSS $\hat{\rho}$; therefore, we can evaluate the system output using our off-line monitor.
If the specification does not have a freeze time operator, 
then the formula is an LTL formula for which the existing monitoring algorithms will be utilized \cite{Rosu01}.
If the specification has a freeze time operator, we first ``instantiate"  \redhighlight{the time variable with the time label of the current sample before formula evaluation. Then, we compute $\bot/\top$ values of the corresponding time constraints.}
When time constraints are evaluated, they will be resolved to $\bot/\top$, and then, the frozen subformula ($x.\varphi(x)$) is 
converted into an LTL formula. Hence, we can apply dynamic programming method \cite{Rosu01} to compute the Boolean value of the frozen subformula.

For each frozen subformula ($x.\varphi(x)$) at each time instance $\tau_i$, we must first precompute the Boolean ($\bot/\top$) value of the corresponding time constraints to transform this frozen subformula into an LTL. 
A two-dimensional matrix $M_{|\phi|\times |\hat{\rho}|}$ with height (number of rows) $|\phi|$ , and width (number of columns) $|\hat{\rho}|$ is created. Here $|\phi|$ denotes the number of subformulas in $\phi$, and $|\hat{\rho}|$ is the number of samples. \redhighlight{Note that row indexing starts from 1 ($\phi\equiv\varphi_1$) up to $|\phi|$ and column indexing starts from 0 ($\rho_0$) up to $|\hat{\rho}|-1$.}

The monitoring table of Example \ref{mtvExmp} is presented in Table \ref{tab:mtvExmp}.
At the beginning, the system outputs corresponding to atomic propositions ($AP=\{a,b\}$) are stored in the rows which belong to the propositions $a$ (row $\varphi_8$) and $b$ (row $\varphi_{13}$) in Table \ref{tab:mtvExmp}.
In Fig. \ref{fig:Tree}, the subformula $\psi_2(x)$ is depicted inside the {\bf white} subtree and $\psi_1(y)$ is depicted inside the {\bf light-gray} subtree.
In the following, we explain the other rows of Table \ref{tab:mtvExmp} and provide a high level overview of the monitoring of $\phi$:

{\bf 1st Run)} We first instantiate time variable $y$ at each sample $i$ with the corresponding timed instance $\tau_i$ to evaluate the Boolean values for the corresponding time constraint $y\le1$ (row $\varphi_{11}$). 
The instantiation transforms $y.\psi_1(y)$ into an LTL formula. Then we compute the Boolean values of $\psi_1(\tau_0)$, $\psi_1(\tau_1)$, $\psi_1(\tau_2)$, $\dots$, $\psi_1(\tau_6)$ from left to right. 
Now the Boolean value of $y.\psi_1(y)$ for each time stamp $\tau_i$ is available for the higher level subtree of the Table \ref{tab:mtvExmp}. 
Therefore, the Boolean values should be copied from row $\varphi_{9}$ to row $\varphi_{6}$.

{\bf 2nd Run)} Given the $\bot/\top$ values of $y.\psi_1(y)$, we can instantiate $x$ at each time stamp $\tau_i$ and modify formula $x.\psi_2(x)$ into an LTL formula. Then we compute the Boolean values of $\psi_2(\tau_0)$, $\psi_2(\tau_1)$, $\psi_2(\tau_2)$, $\dots$, $\psi_2(\tau_6)$ from left to right. Now the Boolean values of $x.\psi_2(x)$ are available for each time stamp $\tau_i$ for the higher subtree.
As a result, the $\bot/\top$ values should be copied from row $\varphi_{3}$ to row $\varphi_{2}$.

{\bf 3rd Run)} The Boolean value of $\Box x.\psi_2(x)$ is computed given the Boolean values of $\psi_2(\tau_i)$ according to the semantics of Always ($\Box$) operator: $$\phi\equiv\sideset{}{_{i=0}^{6}}\bigwedge\psi_2(\tau_i)$$
  \begin{table}[t]
  	\centering
  	\caption{The Monitoring Table of formula $\phi$ of Example \ref{mtvExmp} (Fig. \ref{fig:Tree}) }
  	\fontsize{7}{10}
  	\begin{tabular}{|c|c|c|c|c|c|c|c| }
  		\hline
  		$\varphi_i$(OP) & $\tau_0$& $\tau_1$  & $\tau_2$ &$\tau_3$&$\tau_4$&$\tau_5$&$\tau_6$ \\ \hline
$\varphi_1$($\Box$) \cellcolor{gray!70}& $\{\bot/\top\}$ &   &  &  &  &  &  \\ \hline
 $\varphi_2$($x.$)\cellcolor{gray!70}& $\psi_2(0)$&$\psi_2(\tau_1)$&$\psi_2(\tau_2)$&$\psi_2(\tau_3)$&$\psi_2(\tau_4)$ &$\psi_2(\tau_5)$&$\psi_2(\tau_6)$ \\ \hline \hline
$\varphi_3$($\Diamond$)&$\psi_2(0)$&$\psi_2(\tau_1)$&$\psi_2(\tau_2)$&$\psi_2(\tau_3)$&$\psi_2(\tau_4)$ &$\psi_2(\tau_5)$&$\psi_2(\tau_6)$  \\ \hline
$\varphi_4$($\wedge$) &  & & & && & \\ \hline
  			 
$\varphi_5$($\rightarrow$)&   & & & && &  \\ \hline
$\varphi_6$($y.$) & $\psi_1(0)$\cellcolor{gray!40}& $\psi_1(\tau_1)$\cellcolor{gray!40}&$\psi_1(\tau_2)$\cellcolor{gray!40}&$\psi_1(\tau_3)$\cellcolor{gray!40}&$\psi_1(\tau_4)$\cellcolor{gray!40}& $\psi_1(\tau_5)$ \cellcolor{gray!40} & $\psi_1(\tau_6)$ \cellcolor{gray!40} \\ \hline
$\varphi_7$($x\le1$)&  & & & && & \\ \hline
$\varphi_8$($a$) & & & & && &\\ \hline \hline
$\varphi_9$($\Diamond$)\cellcolor{gray!40}& $\psi_1(0)$\cellcolor{gray!40}&$\psi_1(\tau_1)$\cellcolor{gray!40} &$\psi_1(\tau_2)$ \cellcolor{gray!40}&$\psi_1(\tau_3)$ \cellcolor{gray!40}& $\psi_1(\tau_4)$\cellcolor{gray!40}&$\psi_1(\tau_5)$ \cellcolor{gray!40}&$\psi_1(\tau_6)$\cellcolor{gray!40}\\ \hline
$\varphi_{10}$($\rightarrow$)\cellcolor{gray!40}&& & & && &\\ \hline
$\varphi_{11}$($y\le1$)\cellcolor{gray!40}&& & & && &\\ \hline
$\varphi_{12}$($\neg$)\cellcolor{gray!40}& & & & && & \\ \hline
$\varphi_{13}$($b$)\cellcolor{gray!40}& & & &&& &\\ \hline
  		
  	\end{tabular}
  	\label{tab:mtvExmp}%
  	\vspace{-15pt}
  \end{table}%

\section{TPTL Monitoring Algorithm}
\label{MA}
The algorithms has the main following steps.
\begin{enumerate}
\item For each time variable (frozen subformula) and for each time stamp.
\item Resolve the time constraints into $\bot/\top$ values (This step converts the corresponding frozen subformula into an LTL formula).
\item Compute $\bot/\top$ value of the resulting LTL formula using the dynamic programming algorithm.
\item These $\bot/\top$ values of frozen subformula are used to evaluate the higher level subformulas.
\end{enumerate}
In the following, a detailed description and pseudo code of the proposed algorithm for TPTL monitoring will be explained. 
\subsection{TPTL to LTL Transformation}
\label{TPTL2LTL}

The pseudo code of the monitoring algorithm is provided in Algorithm \ref{alg:off-line} and its main loop has $|V|+1$ iterations where $|V|$ is the number of freeze time variables. Algorithm \ref{alg:off-line} calls Algorithm \ref{alg:DP-LTL} for computing the Boolean value of LTL subformulas. 
The first line of Algorithm \ref{alg:off-line} sets the monitoring table entries of the corresponding atomic propositions, namely the Boolean value of each $p\in AP$ is extracted from the finite state sequence $\hat{\dsig}$.
\redhighlight{In addition, Line 1 sets the monitoring table entries for constant boolean values $\bot/\top$.}
For each time variable $v_k$ (in Line 2), we need to compute the $\bot/\top$ value of the subtree $\theta_k$. The order of $k$ is in such away that the inner most subtree ($\theta_1$) is evaluated first then $\theta_2$, and finally,  $\theta_3$ (See Fig \ref{fig:Tree} for Example \ref{mtvExmp}). This order is crucial for the correctness of the algorithm, because higher level subformulas consider the lower level frozen subformulas as $\bot/\top$. 


To transform the frozen formula into LTL for each sample time $t$ between 0 to $|\hat{\rho}|-1$ (see Line 3), we must first instantiate the time variable $v_k$ to the corresponding time stamp $\tau_t$, then compute the Boolean value of the corresponding time constraint $v_k\sim r$.
The instantiation evaluates the whole constraint row into $\bot/\top$ in Lines 4-13 of Algorithm \ref{alg:off-line}.
The environment is updated based on the time stamp $\tau_t$ and the formula translated into an LTL formula.
Now we use a dynamic programming algorithm based on \cite{Rosu01} to compute the $\bot/\top$ value of the frozen subformula in Lines 14-18. \redhighlight{In Line 15 of Algorithm \ref{alg:off-line}, $\theta_k.max$ ($\theta_k.min$) is the maximum (minimum) index of subformulas in the subtree $\theta_k$.}
In Example \ref{mtvExmp}:\\
1) $\theta_1.min=9$ and $\theta_1.max=13$\\
2) $\theta_2.min=3$ and $\theta_2.max=8$

When the Boolean value of the frozen subformula of $v_k.\psi(v_k)$ ($\theta_k.root$) at time stamp $v_k=\tau_t$ is resolved, this Boolean value is copied to the parent of $\theta_k$ ($\theta_k.parent$) to be used by higher level subformulas (see Line 19 of Algorithm \ref{alg:off-line}).
The loop of Line 3-20 continues for the other time stamps ($\tau_1\dots\tau_{|\hat{\rho}|-1}$) and computes the $\bot/\top$ value of the frozen subformula for each instantiation of $v_k$ to the time stamps $\tau_1\dots\tau_{|\hat{\rho}|-1}$ in this order.
Now we resolved the $\bot/\top$ value of the frozen subformula of $v_k.\psi(v_k)$ for all time stamps.
We continue this process for other time variables (Lines 2-21).

When the Boolean values of the frozen subformulas are resolved for each time variable $v_1\dots v_k\dots v_{|V|}$ in this order, we have an LTL formula for the highest level subformula where it corresponds to subtree $\theta_{|V|+1}$. To compute the $\bot/\top$ value of the highest set of subformulas we run Lines 22-26 of Algorithm \ref{alg:off-line}.
Note that Lines 22-26 are almost identical to Lines 14-18 because the highest set of subformulas is in LTL.
The final value that corresponds to the monitoring trace is stored in table entry $M[1,0]$ and it will be returned to the user. \redhighlight{The table entry $M[1,0]$ contains the Boolean value of the TPTL specification ($\varphi_1$) at sampled index 0.}
\subsection{LTL Monitoring}
Now we explain how to compute the Boolean values of the LTL subtree.
Algorithm \ref{alg:DP-LTL} is based on \cite{Rosu01}, and follows Definition \ref{def:sem4tptl}.
Algorithm \ref{alg:off-line} calls Algorithm \ref{alg:DP-LTL} at each sample $u$.
Algorithm \ref{alg:DP-LTL} has the following 5 cases to compute the Boolean values of the corresponding LTL operators:
\begin{enumerate}
	\item Lines 1-2 for the NOT operation ($\neg$).
	\item Lines 3-4 for the AND operation ($\wedge$).
	\item Lines 5-6 for the OR operation ($\vee$).
	\item Lines 7-12 for the NEXT operation ($\bigcirc$).
	\item Lines 13-19 for the UNTIL operation ($U$).
\end{enumerate}
\begin{algorithm}[t]
	\caption{TPTL Monitor}
	{\bf Input}: $\varphi$, $\hat{\rho}=(\dsig_0,\tau_0)(\dsig_1,\tau_1)\cdots(\dsig_T,\tau_T)$; 
	{\bf Global variables:} $M_{|\varphi|\times |\hat{\rho}|}$; 
	{\bf Output}: $M[1,0]$.
	\label{alg:off-line}

		\hspace {10pt}{\bf   procedure  }{\sc TPTLMonitor}($\varphi, \hat{\rho}$)
		\begin{algorithmic}[1]
			\State Initialize all rows in $M_{|\varphi|\times |\hat{\rho}|}$ corresponding to predicates $\varphi_j\equiv p\in AP$ with $\top/\bot$ value according to $\hat{\dsig}$.
			\For{$k\gets 1\mbox{ to }|V|$}
			\For{$t\gets 0\mbox{ to }|\hat{\rho}|-1$}
			\For{$u\gets t\mbox{ to }|\hat{\rho}|-1$}
			\For{each $\varphi_j\equiv v_k\sim r\in\theta_k$ where\\\hspace{60pt} $j$ is the index of $v_k\sim r$ in $M$}
			\If{$(\tau_u-\tau_t)\sim r$}
			\State $M[j,u]\leftarrow\top$
			\Else
			\State $M[j,u]\leftarrow\bot$
			\EndIf
			\EndFor
			\EndFor
			\For{$u\gets |\hat{\rho}|-1\mbox{ down to }t$}
			\For{$j\gets \theta_k.max\mbox{ down to }\theta_k.min$}
			\State $M[j,u]\leftarrow ComputeLTL(\varphi_j,u,M_{|\varphi|\times |\hat{\rho}|})$
			\EndFor
			\EndFor	
			\State 	$M[\theta_k.parent,t]\leftarrow M[\theta_k.root,t]$ 
			\EndFor 
			\EndFor
			\For{$u\gets |\hat{\rho}|-1\mbox{ down to }0$}
			\For{$j\gets \theta_{|V|+1}.max\mbox{ down to }\theta_{|V|+1}.min$}
			\State $M[j,u]\leftarrow ComputeLTL(\varphi_j,u,M_{|\varphi|\times |\hat{\rho}|})$
			\EndFor
			\EndFor		
						
			\State\Return$M[1,0]$ // Return the value of the first cell/row in $M_{|\varphi|\times |\hat{\rho}|}$ table
		\end{algorithmic}
		\hspace {5pt}{\bf   end procedure}
\end{algorithm}

\begin{algorithm}[t]
	\caption{LTL Monitor}
	{\bf Input}: $\varphi_j,u,M_{|\varphi|\times |\hat{\rho}|}$; 
	{\bf Output}: $M[j,u]$.
	\label{alg:DP-LTL}

	\hspace {10pt}{\bf   procedure  }{\sc ComputeLTL}($\varphi_j,u,M_{|\varphi|\times |\hat{\rho}|}$)
	\begin{algorithmic}[1]
		\If{$\varphi_j\equiv\neg\varphi_m$}
		\State \Return $\neg M[m,u]$
		\ElsIf {$\varphi_j\equiv\varphi_m\wedge\varphi_n$}
		\State \Return $M[m,u]\wedge M[n,u]$
		\ElsIf {$\varphi_j\equiv\varphi_m\vee\varphi_n$}
		\State \Return $M[m,u]\vee M[n,u]$
		\ElsIf {$\varphi_j\equiv\bigcirc\varphi_m$}
		\If{$u=|\hat{\rho}|-1$}
		\State \Return $\bot$
		\Else
		\State \Return $M[m,u+1]$
		\EndIf
		\ElsIf {$\varphi_j\equiv\varphi_mU\varphi_n$}
		\If{$u=|\hat{\rho}|-1$}
		\State \Return $M[n,u]$
		\Else
		\State \Return $M[n,u]\vee(M[m,u]\wedge M[j,u+1])$
		\EndIf
		\EndIf
	\end{algorithmic}
	\hspace {5pt}{\bf   end procedure}
\end{algorithm}
Note that Algorithm \ref{alg:DP-LTL} (ComputeLTL) is $O(1)$ complexity.
Since we can evaluate each frozen subformula ($x.\varphi(x)$) separately because of independent time variables, the time complexity of the algorithm is proportional to the number of time variables and the size of the subformula.
On the other hand, for each time sample we instantiate each time variable to convert the TPTL subformula into an LTL subformula in $O(|\hat{\rho}|)$ then run the LTL monitoring algorithm in $O(|\hat{\rho}|)$. 
As a result, the upper bound on the time complexity of Algorithm \ref{alg:off-line} is $O(|V|\times|\varphi|\times |\hat{\rho}|^2)$, where $|V|$ is the number of time variables, $|\varphi|$ is the number of subformulas, and $|\hat{\rho}|$ is the number of TSS samples.
Both algorithms' correctness proofs are provided in Section \ref{App}.

\subsection{Running example}  
In this section, we utilize our monitoring algorithm to compute the solution for Example \ref{mtvExmp}.
First step of the algorithm is the $\top/\bot$ computation of the frozen subformula $y.\psi_1(y)$ which corresponds to subtree $\theta_1$ and is represented in {\bf light-gray} rows of Tables \ref{tab:mtvExmp} and \ref{tab:mtvExmp2}. 
In Table \ref{tab:mtvExmp2}, when the time value of $y$ is instantiated to 0, then the value of the time constraint $y\le1$ will be resolved for all the samples of $i$ between 0 to $6$ according to the following inequality $\tau_i-0\le1$.
Now $\psi_1(0)$ is transformed into LTL and $\psi_1(0)$ is evaluated, i.e., $\psi_1(0)\equiv\top$ (see row $\varphi_{9}$ column $\tau_0$). 
Then, the time value of $y$ is instantiated to $\tau_1=0.3$ and the value of the time constraint $y\le1$ will be resolved for all the samples of $i$ between 1 to $6$ according to the following inequality $\tau_i-0.3\le1$.
Similarly, $\psi_1(0.3)$ is transformed into LTL and $\psi_1(0.3)$ can be computed, i.e., $\psi_1(0.3)\equiv\top$ (see row $\varphi_{9}$ column $\tau_1$).
We continue the computation of $\psi_1(\tau_i)$ with the following instantiation $\tau_2=0.7,\dots,\tau_6=1.9$ similar to $\tau_0$.
Now $\bot/\top$ values of the frozen subformula $y.\psi_1(y)$ for each time stamp $\tau_i$ are available in row $\varphi_{9}$ of Table \ref{tab:mtvExmp2}.

The Boolean values of subtree $\theta_1$ should be available for higher level subformulas.
Therefore, the row $\varphi_{9}$ will be copied to row $\varphi_{6}$ (in Table \ref{tab:mtvExmp2} both rows have the same color).
Now we can continue the second run of the algorithm.
The $\top/\bot$ computation of the frozen subformula $x.\psi_2(x)$ which corresponds to subtree $\theta_2$ is represented in {\bf white} rows of Table \ref{tab:mtvExmp} and \ref{tab:mtvExmp2}. 
In Table \ref{tab:mtvExmp2}, the time value of $x$ is instantiated to 0, then the value of $\psi_2(0)$ is computed, i.e., $\psi_2(0)\equiv\top$ (see row $\varphi_{3}$ column $\tau_0$). 
Now, the time value of $x$ is instantiated to $\tau_1=0.3$ and the value of $\psi_2(0.3)$ is computed $\psi_2(0.3)\equiv\top$ (see row $\varphi_{3}$ column $\tau_1$). 
We continue the computation of $\psi_2(\tau_i)$ similarly with $\tau_2=0.7\dots\tau_6=1.9$. 
Now the $\bot/\top$ values of the frozen subformula $x.\psi_2(x)$ for each time stamp $\tau_i$ are available in row $\varphi_{3}$ of Table \ref{tab:mtvExmp2}.
Since the Boolean values of subtree $\theta_2$ should be available for higher level subformulas, the row $\varphi_{3}$ is copied to row $\varphi_{2}$.
Finally, we compute $\phi=\Box x.\psi_2(x)$ using Lines 22-26 of Algorithm \ref{alg:off-line} which corresponds to following: $\phi=\sideset{}{_{i=0}^{6}}\bigwedge\psi_2(\tau_i)\equiv\bot$

  \begin{table*}[t]
  	\centering
  	\caption{Computing the Boolean values for $\phi=\Box x.\psi_2(x)$. Boolean values correspond to the final snapshot of Monitoring Table.}
  	\fontsize{7.5}{10}
  	\begin{tabular}{|c|c|c|c|c|c|c|c|c| }
  		\hline
  		$\varphi_i$	&	subformula &$\tau_0=0$&$\tau_1=0.3$&$\tau_2=0.7$&$\tau_3=1.0$&$\tau_4=1.1$&$\tau_5=1.5$& $\tau_6=1.9$ \\ \hline
  		$\varphi_1$\cellcolor{gray!70}&$\phi=\Box x.\psi_2(x)$\cellcolor{gray!70}& $\bot$ & $\bot$ &$\bot$ & $\bot$ & $\bot$ & $\bot$ & $\bot$ \\ \hline
  		$\varphi_2$ \cellcolor{gray!70}&$x.\psi_2(x)\equiv x.\Diamond((x\le1\rightarrow a)\wedge y.\psi_1(y))$\cellcolor{gray!70}& $\psi_2(0)\equiv\top$&$\psi_2(\tau_1)\equiv\top$&$\psi_2(\tau_2)\equiv\top$&$\psi_2(\tau_3)\equiv\top$&$\psi_2(\tau_4)\equiv\bot$& $\psi_2(\tau_5)\equiv\bot$&$\psi_2(\tau_6)\equiv\bot$ \\  \hline \hline
  		$\varphi_3$&$\Diamond((x\le1\rightarrow a)\wedge y.\psi_1(y))$&$\top$&$\top$&$\top$&$\top$&$\bot$&$\bot$&$\bot$\\ \hline
  		$\varphi_4$&$(x\le1\rightarrow a)\wedge y.\psi_1(y)$& $\bot$ & $\bot$ &  $\top$ & $\top$ & $\bot$ & $\bot$ & $\bot$ \\ \hline
  		$\varphi_5$&$x\le1\rightarrow a$	& $\bot$ & $\bot$ &  $\top$ & $\top$ & $\top$ & $\bot$ & $\bot$ \\ \hline
  		$\varphi_6$	& $y.\psi_1(y)\equiv y.\Diamond(y\le1\rightarrow\neg b)$  & $\psi_1(0)\equiv\top$\cellcolor{gray!40} &  $\psi_1(\tau_1)\equiv\top$ \cellcolor{gray!40}  & \ $\psi_1(\tau_2)\equiv\top$  \cellcolor{gray!40} &  $\psi_1(\tau_3)\equiv\top$ \cellcolor{gray!40}  &  $\psi_1(\tau_4)\equiv\bot$ \cellcolor{gray!40}  & $\psi_1(\tau_5)\equiv\bot$  \cellcolor{gray!40} & $\psi_1(\tau_6)\equiv\bot$ \cellcolor{gray!40} \\  \hline
  		$\varphi_7$&$x\le1$& $\top$ & $\top$ &  $\top$ & $\top$ & $\top$ & $\top$ & $\top$\\ \hline  		
  		$\varphi_8$	&$a$  & $\bot$ & $\bot$  & $\top$  & $\top$  &  $\top$  &  $\bot$  &  $\bot$ \\\hline \hline
  		$\varphi_9$\cellcolor{gray!40}	&$\Diamond(y\le1\rightarrow\neg b)$\cellcolor{gray!40}	 &$\top$\cellcolor{gray!40}&$\top$\cellcolor{gray!40}&$\top$\cellcolor{gray!40}&$\top$\cellcolor{gray!40}&$\bot$\cellcolor{gray!40}&$\bot$\cellcolor{gray!40}&$\bot$\cellcolor{gray!40} \\  \hline
  		$\varphi_{10}$\cellcolor{gray!40}	&$y\le1\rightarrow\neg b$\cellcolor{gray!40} & $\top$ & $\top$ & $\bot$ & $\top$ &  $\bot$  & $\bot$  &  $\bot$ \\ \hline
  		$\varphi_{11}$\cellcolor{gray!40}	&$y\le1$\cellcolor{gray!40}	& $\top$ & $\top$ &  $\top$   &  $\top$  & $\top$    & $\top$  & $\top$ \\ \hline
  		$\varphi_{12}$\cellcolor{gray!40}	&$\neg b$\cellcolor{gray!40}	& $\top$ & $\top$  & $\bot$  & $\top$  &  $\bot$ & $\bot$  & $\bot$ \\ \hline
  		$\varphi_{13}$\cellcolor{gray!40}	&	$b$ \cellcolor{gray!40} &  $\bot$ &  $\bot$  &  $\top$ &  $\bot$  & $\top$  & $\top$  & $\top$ \\ \hline
  		
  	\end{tabular}
  	\label{tab:mtvExmp2}%
  	\vspace{-12pt}
  \end{table*}%

  		

\section{Experiments}
  
An implementation of our TPTL monitoring algorithm is provided in the \staliro testing framework \cite{hoxhatowards}.
\staliro is a Matlab toolbox that uses stochastic techniques to 
find initial states and inputs to Simulink models which result in trajectories that falsify MTL formulas.
With our TPTL off-line monitoring algorithm, \staliro can evaluate specifications that are more expressive than MTL.
\subsection{Runtime Analysis}
We measured the runtime of our TPTL monitoring algorithm using the \staliro toolbox.
The system under test was the Automatic Transmission (AT) model provided by Mathworks as a Simulink demo \cite{AT}.
We introduced a few modifications to the model to make it compatible with the \staliro framework, which are explained in \cite{HoxhaAF14arch}. 
AT has two inputs of Throttle and Brake.
The outputs contain two real-valued traces: the rotational speed of the engine $\omega$ and the speed of the vehicle 
$v$. 
In addition, the outputs contain one discrete-valued trace $gear$ with four possible values. 

To provide TPTL specifications, we defined four atomic propositions corresponding to the following predicates:\\
{\bf 1)} $a_1\equiv(\omega\ge4500)$: ``rotational  speed of the engine $\ge$ 4500''\\
{\bf 2)} $a_2\equiv(\omega\le1500)$: ``rotational speed of the engine $\le$ 1500''\\
{\bf 3)} $a_3\equiv(v\ge40)$: ``speed of the vehicle $\ge$ 40''\\
{\bf 4)} $a_4\equiv(v\le120)$: ``speed of the vehicle $\le$ 120''\\
Note that these predicates are chosen to be non-trivial and have meaning in the CPS context.
The TPTL formulas are generated based on typical safety reactive response specifications.
We generated these TPTL formula patterns to check the runtime with respect to: 1) Size of system trace 2) Number of temporal operators 3) Number of time variables.
 
We created 18 TPTL formulas that cannot be expressed in MTL. 
All the specifications have the reactive response pattern: $\Box (a_1 \rightarrow \psi)$
where $\psi$ is categorized in two groups:
\begin{enumerate}
	\item EA group ($\psi_{EA}$): contains Eventually/Always specifications with 2, 4 and 8 temporal operators.
	\item UR group ($\psi_{UR}$): contains Until/Release specifications with 2, 4 and 8 temporal operators.
\end{enumerate}
We first chose a $\psi$ specification in LTL from Table \ref{tab:tptl_specs} column (LTL template).
In Table \ref{tab:tptl_specs}, column (\#) represents the number of temporal operators for each LTL template.
Then, we added time variables to create a TPTL specification.
The last column in Table \ref{tab:tptl_specs} represents the number of TPTL formulas that we created by adding time constraints on $\psi$.  
The time variables that we add to $\psi$ correspond to individual temporal operators.
In this case, for $\psi_{EA2}$ we create two TPTL formulas with one and two time variables respectively given as $\phi_1$ and $\phi_2$:
\begin{enumerate}
	\item [EA] $\phi_1=\Box (a_1 \rightarrow x.\Diamond(a_2\wedge\Box (a_3\vee a_4\wedge C_x )))$
	\item [EA] $\phi_2=\Box (a_1 \rightarrow x.\Diamond(a_2\wedge C_x \wedge y.\Box (a_3\vee a_4\wedge C_y)))$
\end{enumerate}
where $C_x$ and $C_y$ are the corresponding time constraints for $x$ and $y$.
Similarly for $\psi_{UR2}$ we created two TPTL formulas with one and two time variables respectively given as $\phi_1$ and $\phi_2$:
\begin{enumerate}
	\item [UR] $\phi_1=\Box (a_1 \rightarrow x.(a_2 U ( a_3 R (a_4\wedge C_x) )))$
	\item [UR] $\phi_2=\Box (a_1 \rightarrow x.(a_2 U a_4\wedge C_x \wedge y.( a_3 R (a_4\wedge C_y)))$
\end{enumerate}
We used a similar method to generate $\phi_3$ with one time variable, $\phi_4$ with two time variables, and $\phi_5$ with four time variables based on $\psi_{EA4}$ and $\psi_{UR4}$ with the total number of six TPTL formulas.
Finally, we create eight TPTL formulas based on $\psi_{EA8}$ and $\psi_{UR8}$. 
These formulas are $\phi_6$, $\phi_7$, $\phi_8$, $\phi_9$ and they are represented in Table \ref{tab:overheadRes}.
Our experiments were conducted on a 64-bit Intel Xeon CPU (2.5GHz) with 64-GB RAM and Windows Server 2012.
We used Matlab 2015a and Microsoft Visual C++ 2013 Professional to compile our algorithms' code (in C) using the Matlab mex compiler.

The runtime is provided in Table \ref{tab:overheadRes}.
Each row considers two TPTL formulas in EA or UR configuration. For example, the first column $\phi_1$ represents $\Box (a_1 \rightarrow x.\Diamond(a_2\wedge\Box (a_3\vee a_4\wedge C_x )))$ and $\Box (a_1 \rightarrow x.(a_2 U ( a_3 R (a_4\wedge C_x) )))$ in EA and UR configurations, respectively.
In Table \ref{tab:overheadRes} the second column ($\#$) represents the number of temporal operators in the corresponding frozen subformula, namely, the number of of temporal operators in $\psi_{EA\#}$ or $\psi_{UR\#}$.
The third column ($|V|$) in Table \ref{tab:overheadRes} represents the number of time variables in $\psi_{EA\#}$ or $\psi_{UR\#}$. 

We tested our algorithm with the execution traces of the length 1000, 2000, and 10000.
For each TPTL formula, we tested our algorithm 100 times where the AT's throttle input is provided by random signal generator (without brake).
We reported the mean value (in {\bf Bold}) and variance of the algorithm's runtime in Table \ref{tab:overheadRes}.
It can be seen that when the length of the trace doubles from $|\hat{\rho}|$=1,000 to $|\hat{\rho}|$=2,000 , the runtime  quadruples (see {\bf Mean} values in Table \ref{tab:overheadRes}).
Similarly, when the length of trace increases ten times from $|\hat{\rho}|$=1,000 to $|\hat{\rho}|$=10,000 the runtime  increased 100 times (see {\bf Mean} values in Table \ref{tab:overheadRes}).
Now, consider the mean values of $\phi_1$ and $\phi_2$. The number of time variables in $\phi_1$ is one and in $\phi_2$ is two.
It can be seen that mean values of  $\phi_2$ are twice as those of $\phi_1$.
Similarly, comparing $\phi_3$ and $\phi_4$ and $\phi_5$ shows that the runtime is proportional to the number of time variables. 
Finally, comparing rows $\phi_1$ and $\phi_3$ and $\phi_6$ shows that the runtime relates to the number of temporal operators.
\redhighlight{The experimental results indicate that the runtime behaves as expected, considering that our algorithm is in $O(|V|\times|\varphi|\times |\hat{\rho}|^2)$.}
\begin{table}[t]
	\centering
	\caption{Specifications of $\psi$ before adding time variables.}{
	\begin{tabular}{| c | l | c | c |}
	\hline
	LTL	& \# & LTL template & TPTLs \\ \hline 
	$\psi_{EA2}$ & 2 & $\Diamond(a_2\wedge\Box (a_3\vee a_4)$ & 2\\
	\hline 
	$\psi_{EA4}$ & 4 &  $\Diamond ( a_2 \wedge \Box (  a_3 \vee  a_4 \wedge \psi_{EA2})$ & 3 \\ \hline 
	$\psi_{EA8}$ & 8 & $\Diamond ( a_2 \wedge \Box (  a_3 \vee  a_4 \wedge \Diamond ( a_2 \wedge \Box (  a_3 \vee  a_4 \wedge \psi_{EA4}))))$ & 4 \\  \hline 
	$\psi_{UR2}$ & 2 & $ a_2 U ( a_3 R a_4 )$ & 2 \\  \hline 
	$\psi_{UR4}$ & 4 & $ a_2 U ( a_3 R ( a_4 \wedge \psi_{UR2} ))$ & 3 \\  \hline 
	$\psi_{UR8}$ & 8 & $  a_2 U ( a_3 R ( a_4 \wedge ( a_2 U ( a_3 R ( a_4 \wedge \psi_{UR4} )))))$ & 4 \\  \hline 

	\end{tabular}}
      \label{tab:tptl_specs}
		    \vspace{-15pt}
\end{table}

\begin{table*}[t]
	\centering
	\caption{The runtime of Monitoring Algorithm for 18 TPTL formulas. All the values are in seconds.  }
	\begin{tabular}{|c|c|c|c|c|c|c||c|c|c|c||c|c|c|c|}
		\cline{4-15}
		\multicolumn{3}{l}{} & \multicolumn{4}{|c||}{$|\hat{\rho}|$=1,000} & \multicolumn{4}{c||}{$|\hat{\rho}|$=2,000} & \multicolumn{4}{c|}{$|\hat{\rho}|$=10,000} \\
		\cline{4-15}
		\multicolumn{3}{l}{} & \multicolumn{2}{|c|}{EA ($\psi_{EA\#}$)}     &\multicolumn{2}{c||}{UR ($\psi_{UR\#}$)}     &\multicolumn{2}{c|}{EA ($\psi_{EA\#}$)}      & \multicolumn{2}{c||}{UR ($\psi_{UR\#}$)}     & \multicolumn{2}{c|}{EA ($\psi_{EA\#}$)}     & \multicolumn{2}{c|}{UR ($\psi_{UR\#}$)}  \\ \hline
		$\phi$ & \# &  $|V|$     & {\bf Mean}  & Var. & {\bf Mean}  & Var. & {\bf Mean}  & Var. & {\bf Mean}  & Var. & {\bf Mean}  & Var. & {\bf Mean}  & Var.  \\ \hline\hline           
		\multicolumn{1}{|l|}{$\phi_1$}   &2 &  1    & {\bf 0.077}  & 0.0002 & {\bf 0.064}  & 0.000  &  {\bf 0.326}  & 0.001  &  {\bf 0.250} & 0.0013  &  {\bf 8.512} & 0.066  &  {\bf 6.427} & 0.068 \\\hline\hline
		\multicolumn{1}{|l|}{$\phi_2$}   &2 &  2    &   {\bf 0.151} & 0.0005 &  {\bf 0.137} & 0.0003 &  {\bf 0.5887} & 0.0018 &  {\bf 0.551}  & 0.002  &  {\bf 14.31}  & 0.191 &  {\bf 13.67} & 0.175 \\ \hline\hline
		\multicolumn{1}{|l|}{$\phi_3$}   &4 &  1   &  {\bf 0.142} &  0.0003  &  {\bf 0.097}  & 0.0001 &  {\bf 0.5885} & 0.002 &  {\bf 0.382}  & 0.002  &  {\bf 15.33} &  0.232  &  {\bf 10.46} & 0.154 \\ \hline\hline
		\multicolumn{1}{|l|}{$\phi_4$}   &4 &  2   &  {\bf 0.205} & 0.0003  &  {\bf 0.15}  & 0.0002 &  {\bf 0.871} & 0.0032 &  {\bf 0.604}  & 0.002  & {\bf 22.9} &  0.344  & {\bf 16.35} & 0.24\\\hline\hline
		\multicolumn{1}{|l|}{$\phi_5$}   &4 &  4    & {\bf 0.417} & 0.0012  & {\bf 0.38}  & 0.0004 & {\bf 1.721} & 0.0058 & {\bf 1.558}  & 0.007  & {\bf 46.25} & 7.08   & {\bf 41.2} & 1.077 \\\hline\hline
		\multicolumn{1}{|l|}{$\phi_6$}   & 8& 1    & {\bf 0.227} & 0.0001  & {\bf 0.154}  & 0.0002 & {\bf 0.948} & 0.005 & {\bf 0.552} & 0.0046 & {\bf 30.27} &  9.708  & {\bf 17.01} & 2.184 \\	\hline\hline	
		\multicolumn{1}{|l|}{$\phi_7$}    &8 &  2   & {\bf 0.367} & 0.025  & {\bf 0.235}  & 0.0011 & {\bf 1.474} & 0.0078 & {\bf 1.023}  & 0.0137  & {\bf 41.59} &  2.17  & {\bf 26.95} & 2.204 \\\hline\hline
		\multicolumn{1}{|l|}{$\phi_8$}    & 8&  4   & {\bf 0.533} & 0.0042  & {\bf 0.437}  & 0.0013 & {\bf 2.26} & 0.024 & {\bf 1.751}  & 0.0115  & {\bf 66.13} & 34.36   & {\bf 48.95} & 8.857 \\\hline\hline
		\multicolumn{1}{|l|}{$\phi_9$}   &8 &   8   & {\bf 1.145} & 0.025  & {\bf 1.093} & 0.0066 & {\bf 4.9} & 0.0391 & {\bf 4.346} & 0.1413 & {\bf 137} & 220   & {\bf 124.6} & 184 \\
		\hline
	\end{tabular}%
	\label{tab:overheadRes}%
	\vspace{-10pt}
\end{table*}%
\vspace{-10pt}

\subsection{Case Study}
In this section, we consider CPS requirements which are impossible to formalize in MTL \cite{BouyerCM10}, but we formalize them in TPTL, very easily.
The ultimate goal is to run the testing algorithm on these requirements.
Our TPTL monitoring algorithm is provided as add-on to the \staliro testing framework.  
\staliro searches for counterexamples to MTL properties through global minimization of a robustness metric \cite{FainekosP09tcs}.
The robustness of an MTL formula $\varphi$ is a value that measures how far is the trace from the satisfaction/falsification of $\varphi$. 
This measure is an extension of Boolean values ($\top/\bot$) for representing satisfaction or falsification. 
A positive robustness value means that the trace satisfies the property and a negative value means that the property is not satisfied.
The stochastic search then returns the simulation trace with the smallest robustness value that was found.

To falsify safety requirements in TPTL which are more expressive than MTL, we should use our proposed TPTL monitor that can handle those specifications. 
Now let us consider the Automatic Transmission (AT) system. 
It contains the discrete output $gear$ signal with four possible values 
($gear=1$, ..., $gear=4$) which indicate the current gear in the auto-transmission controller.
We use four atomic propositions $g_1,g_2,g_3,g_4$ for each possible gear value, where $(gear=i)\equiv g_i$. 
Then we define three up-shifting events as follows:\\
1) $e_1=g_1 \wedge \bigcirc g_2$ means shift from gear one to gear two.\\
2) $e_2=g_2 \wedge \bigcirc g_3$ means shift from gear two to gear three.\\
3) $e_2=g_3 \wedge \bigcirc g_4$ means shift from gear three to gear four.\\\\
In CPS, it is possible that we need to specify the safety requirement about three or more events in sequence, but the time difference between the first and last event happening should be of importance.
In general, these types of specification are impossible to represent in MTL. 
We provide two very succinct TPTL specifications that can formalize these challenging requirements.

The first requirement is as follows:\\``{\it Always if $e_1$ happens, then if $e_2$ happens in future and if $e_3$ happens in future after $e_2$, then the duration between $e_1$ and $e_3$ should be equal or more than 8.}"

This specification is formalized in the following formula:
$$\Phi_1=\Box z.( e_1 \rightarrow \Box ( e_2 \rightarrow \Box( e_3 \rightarrow z\ge8 )) )$$
\staliro successfully falsified $\Phi_1$ which is represented in Fig. \ref{fig:TPTL1}.
In Fig. \ref{fig:TPTL1} the Throttle, Break, and Gear trajectory of the corresponding falsification is presented. 
It can be seen that the duration between $e_1$ and $e_3$ is less that 8. Its actual value is $8.4-1.72=6.68 < 8$.
\begin{figure}[!t]
	\centering 
	\includegraphics[width=\columnwidth]{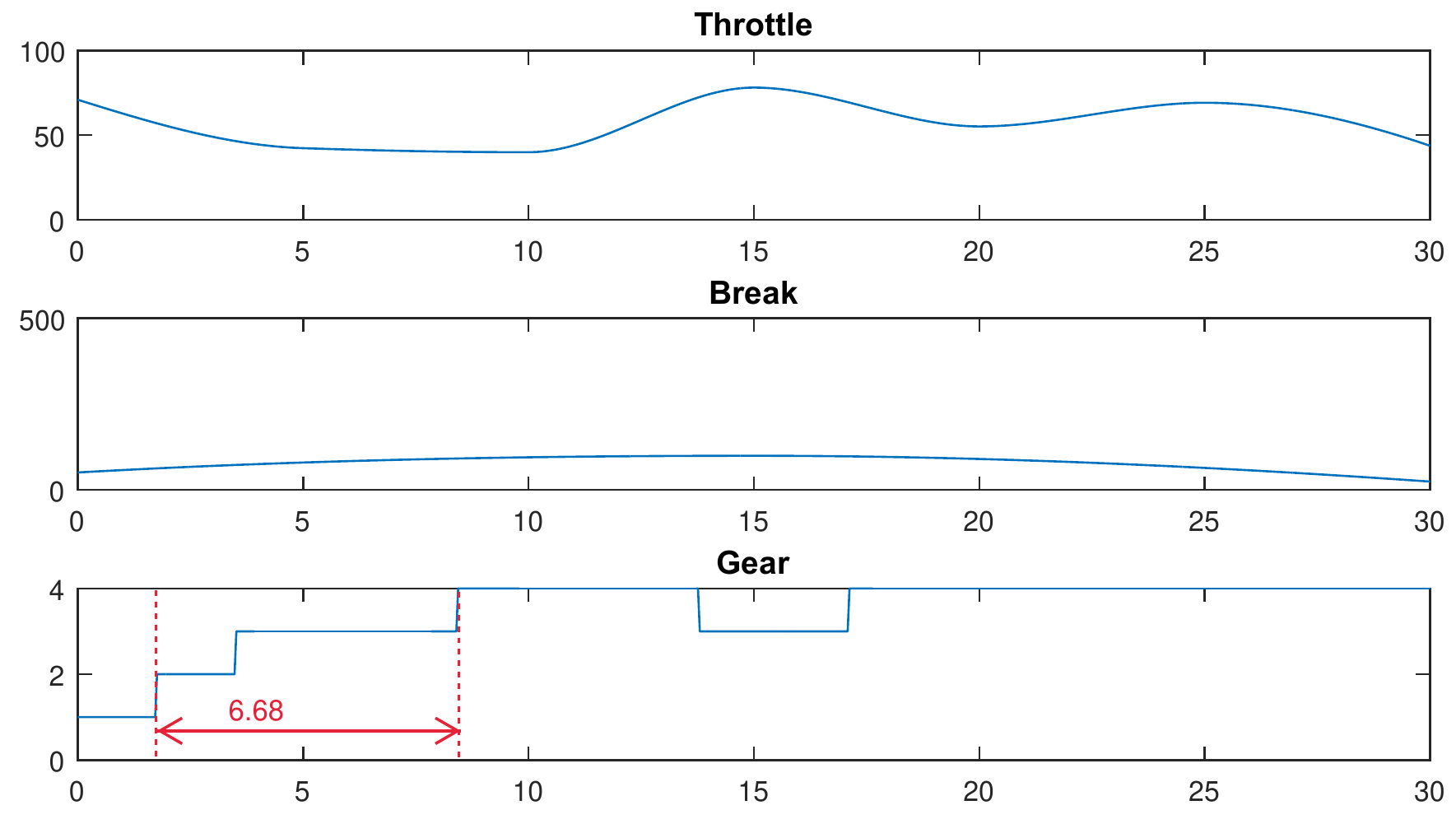}
	\vspace{-15pt}
	\caption{Falsification of $\Phi_1$ using \staliro. The duration between $e_1$ and $e_3$ is less than 8 seconds.}
	\label{fig:TPTL1}
	\vspace{-10pt} 
\end{figure}

The second requirement is as follows:\\
``{\it Always if $e_1$ happens, then $e_2$ should happen in future, and $e_3$ should happen in future after $e_2$, and the duration between $e_1$ and $e_3$ should be equal or less than 12.}"

This specification is formalized by the following formula:
$$\Phi_2=\Box z.( e_1 \rightarrow \Diamond ( e_2 \wedge \Diamond( e_3 \wedge z\le12 )) )$$

\begin{figure}[!t]
	\centering 
	\includegraphics[width=\columnwidth]{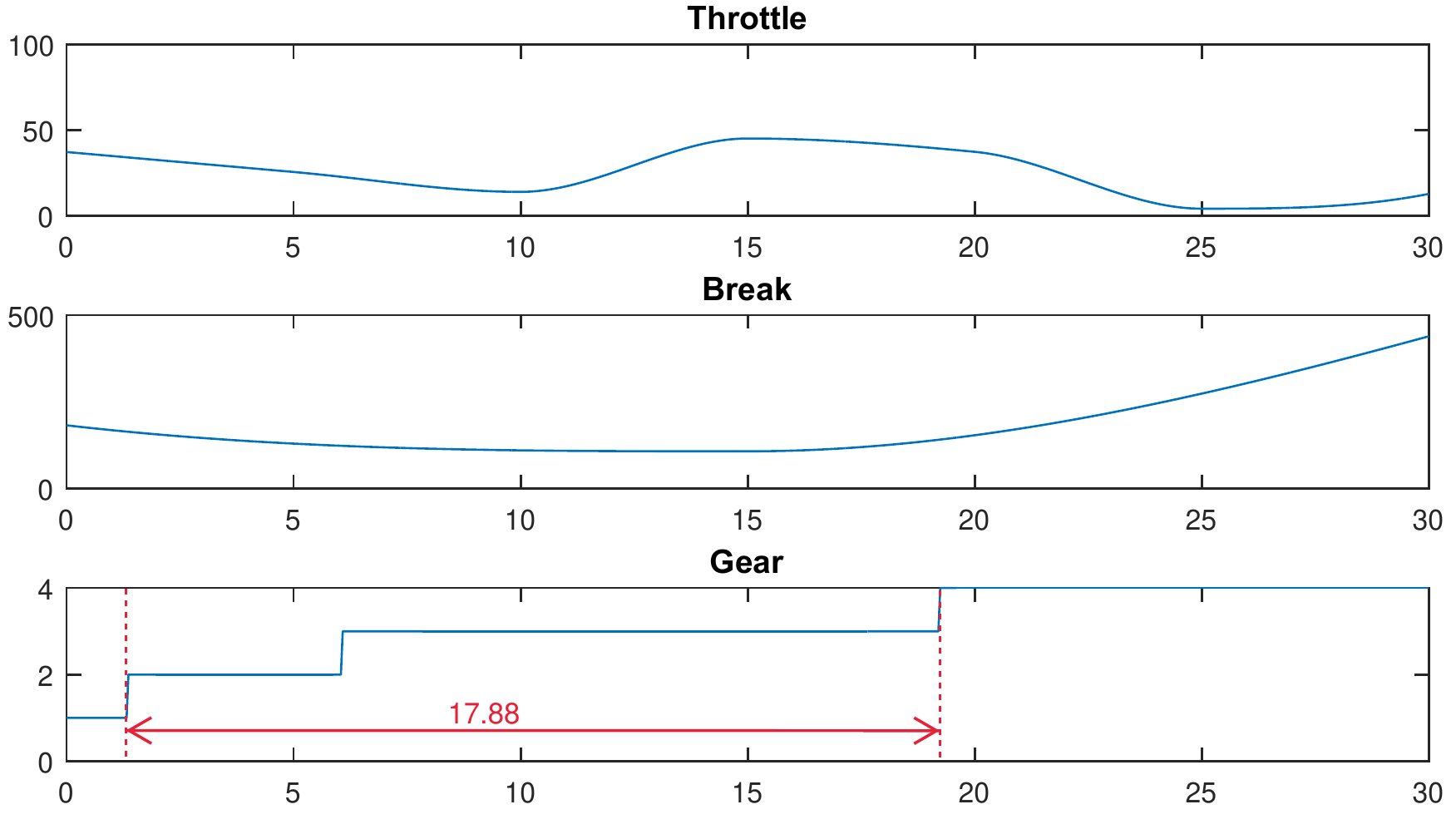}
	\vspace{-15pt}
	\caption{Falsification of $\Phi_2$ using \staliro. The duration between $e_1$ and $e_3$ is more than 12 seconds.}
	\label{fig:TPTL2}
	\vspace{-15pt} 
\end{figure}

In Fig. \ref{fig:TPTL2} the Throttle, Break, and Gear trajectories of the falsification of $\Phi_2$ are represented. 
It can be seen that the duration between $e_1$ and $e_3$ is more than 12, its actual value is $19.2-1.32=17.88>12$.
This case study shows that \staliro can be used for the falsification problem of challenging TPTL requirements.
The method we propose in this work opens the possibility for CPS off-line monitoring of very complex specifications in TPTL using an efficient algorithm.
  
\section{Conclusions and Future works}
In this paper, we provide an efficient polynomial time algorithm for a practical subset of TPTL specifications.
We show that very complex specifications can be succinctly represented in this TPTL subset.
In addition, we can combine full TPTL with a bounded number of time variables with our suggested algorithm to test the specifications that have an arbitrary number of independent time variables and full TPTL with limited number of time variables.
Finally, our method can help CPS developers to efficiently test requirements that cannot be expressed in MTL.

\paragraph*{Acknowledgments}

This research was partially funded by NSF awards CNS-1350420 and CNS-1319560.

\bibliographystyle{abbrv}
\bibliography{adel}

\section{Appendix}
\label{App}
In this section, we will prove the correctness of Algorithms \ref{alg:off-line} and \ref{alg:DP-LTL}. 
Our method first transforms the TPTL formula into LTL formula using Algorithm \ref{alg:off-line}.
Then it uses the dynamic programming method for monitoring LTL using Algorithm \ref{alg:DP-LTL}.
\subsection{Proof of the correctness of Algorithm \ref{alg:off-line}}
\begin{theorem}
	Given an encapsulated TPTL formula $\varphi$, and a finite TSS $\hat{\rho}$, after the execution of Algorithm \ref{alg:off-line} the returned value is:$$M[1,0]=\top\mbox{ iff }(\hat{\rho},0,\mbox{{\bf 0}})\models\varphi$$
\end{theorem}
To prove this theorem, we must show that the Boolean value of the subformulas that are computed using Algorithm \ref{alg:off-line}, follows the TPTL semantics in Definition \ref{def:sem4tptl}. 
Since Algorithm \ref{alg:off-line} does not evaluate propositional and temporal operators, their corresponding proof will be provided in Section \ref{proofLTL}.

According to the TPTL semantics in Definition \ref{def:sem4tptl}, for each freeze time operation $x.\varphi(x)$, and for each time stamp $\tau_i$ we must instantiate the time variable $x$ with the value of $\tau_i$.
This instantiation enables us to evaluate time constraints and transform TPTL to LTL.
The loop of Lines 2-21 is the main loop of Algorithm \ref{alg:off-line} which instantiates each variable $v_k$ with each time sample $\tau_t$ in Line 3. 
\begin{lemma}
	\label{LI}
The loop invariant of Algorithm \ref{alg:off-line} is as follows:
$$\forall j,k,t\mbox{ where }\varphi_j\equiv v_k.\varphi_i, 0\le t<|\hat{\rho}|:$$
$$M[j,t]=\top\mbox{ iff }(\hat{\rho},t,\varepsilon)\models v_k.\varphi_i$$
\end{lemma}
We use induction to prove the loop invariant of Algorithm \ref{alg:off-line}.

{\bf Base:} If $|V|=0$, then formula is in LTL and algorithm does not enter the to loop of Lines 2-21 (only executes Lines 22-26). The proof of LTL is provided in Section \ref{proofLTL}.

{\bf Induction Hypothesis:} We assume for all $v_l$, where $l<k$ the invariant holds. In other words
$$\forall j,l<k,t\mbox{ where }\varphi_j\equiv v_l.\varphi_i, 0\le t<|\hat{\rho}|:$$
$$M[j,t]=M[\theta_l.parent,t]=\top\mbox{ iff }(\hat{\rho},t,\varepsilon)\models v_l.\varphi_i$$

{\bf Induction Step:} To show the correctness for the case of $v_k$, we prove that Algorithm \ref{alg:off-line} correctly transform TPTL into LTL. Then we apply the correctness of LTL (See Section \ref{proofLTL}) to establish the correctness of invariant considering $v_k$.
Thus, we consider two cases that instantiate and evaluate $v_k$ and show that Algorithm \ref{alg:off-line} follows the semantics in Definition \ref{def:sem4tptl}. According to I.H. and since time variables are independent, we can correctly consider frozen subformulas of $\varphi_i$ as $\top/\bot$. As a result, we will conclude that $\varphi_i$ is in LTL.\\
\\{\bf Case of $v_k.\varphi_i$:}\\
Consider the semantics of the freeze operator in Definition \ref{def:sem4tptl}:
$$(\hat{\rho},t,\varepsilon)\models v_k.\varphi_i\mbox{ iff }(\hat{\rho},t,\varepsilon[v_k:=\tau_t])\models\varphi_i$$
According to this semantics, the freeze operation ``$v_k.$'' first assigns a new value to the variable $(v_k:=\tau_t)$. 
Then the $\top/\bot$ value of $v_k.\varphi_i\equiv\varphi_j$ will be resolved to the same $\top/\bot$ value of $\varphi_i$ (with the new environment update). 
Therefore, for each variable assignment $(v_k:=\tau_t)$, we first update the environment variables (Algorithm \ref{alg:off-line}, Line 3), and then copy the $\varphi_i$'s $\top/\bot$ value into $v_k.\varphi_i$'s corresponding row (Algorithm \ref{alg:off-line}, Line 19).

Since each time variable $v_k$ is independent, we create the subtree (set) $\theta_k$ corresponding to the subformulas of  $v_k.\varphi_i(v_k)$ (see Section \ref{datastructure}). 
To evaluate $v_k.\varphi_i(v_k)$, we must first instantiate variable $v_k$ for each time stamp $\tau_0\dots\tau_{|\hat{\rho}|-1}$.
This instantiation is considered in Line 2 of Algorithm \ref{alg:off-line} for time variable $v_k$ and for each sample of time $0\dots (|\hat{\rho}|-1)$ in Line 3 of Algorithm \ref{alg:off-line}.
Now we must copy the resulting $\top/\bot$ value from $\varphi_i$ back to $v_k.\varphi_i$.
The row corresponding to $\theta_k.root$ contains the $\top/\bot$ value of $\varphi_i$ which is the root of $\theta_k$ subtree.
This values must be copied to the row $\theta_k.parent$ which is the parent of subtree $\theta_k$ and it corresponds to $\varphi_j$ (Algorithm \ref{alg:off-line}, Line 19).
\\
\\{\bf Case of $v_k\sim r$:}\\
Consider the semantics of time constraints in Definition \ref{def:sem4tptl}:
$$(\hat{\rho},u,\varepsilon)\models v_k\sim r \mbox{ iff } (\tau_u-\varepsilon(v_k))\sim r$$
In the above semantics, $\varepsilon(v_k)$ corresponds to the frozen value of the time variable $v_k$ (environment of $v_k$).
In the previous case for $v_k.\varphi_i$, we mentioned that we should instantiate $v_k$ at each time stamp $\tau_0\dots\tau_{|\hat{\rho}|-1}$.
According to semantics in Definition \ref{def:sem4tptl}, each freeze operator assigns the environment variable for the current and future samples of time $t$:\\
$$(\hat{\rho},t,\varepsilon)\models v_k.\varphi_i\mbox{ iff }(\hat{\rho},t,\varepsilon[v_k:=\tau_t])\models\varphi_i$$
Which means that the environment updates $\varepsilon[x:=\tau_t]$ are observable for the current and the future samples ($t\le u$).
Therefore, after we instantiated variable $v_k$ at each time stamp $\tau_t$, the environment update will affect all the samples $u$ between $t\le u\le |\hat{\rho}|-1$.
As a result, the time constraint $v_k\sim r$ must be updated for all future samples of $t\le u\le |\hat{\rho}|-1$ for $\varepsilon[v_k:=\tau_t]$ instantiation. 

Lines 4-13 of Algorithm \ref{alg:off-line} follow the above discussion.
Namely, for time variable $v_k$, we instantiate each time stamp $\tau_t$ (Line 3), the time constraints of current/future samples are evaluated according to the frozen time stamp $\tau_t$.
Actual evaluation happens in the Line 7 of Algorithm \ref{alg:off-line}, where $(\tau_u-\tau_t)\sim r$ follows the semantic $(\tau_u-\varepsilon(v_k))\sim r$ for each environment assignment of $\varepsilon[v_k:=\tau_t]$.
Lines 14-18 of Algorithm \ref{alg:off-line} will evaluate the LTL formula $\varphi_i(\tau_t)$.

So far, we transformed TPTL $v_k.\varphi_i(v_k)$ into LTL $\varphi_i(\tau_t)$ for each time stamp $\tau_t$.
Now we can prove that the loop invariant of Algorithm \ref{alg:off-line} holds for $v_k$.
\begin{proof}
	We will prove the Induction Step by assuming the correctness of LTL formula $\varphi_i$ according to Section \ref{proofLTL}:
	$$\forall i,t,\varepsilon\mbox{ where }\varphi_i\subset LTL, 0\le t<|\hat{\rho}|$$
	$$M[i,t]=\top\mbox{ iff }(\hat{\rho},t,\varepsilon)\models\varphi_i$$
	Since for each $\theta_k$, $i=\theta_k.root$ is the index of the highest LTL, $M[\theta_k.root,t]$ will also contain the correct $\top/\bot$ value, therefore\\
	$M[i,t]=M[\theta_k.root,t]=\top\mbox{ iff }(\hat{\rho},t,\varepsilon)\models\varphi_i(v_k=\tau_t)$ iff $(\hat{\rho},t,\varepsilon[v_k:=\tau_t])\models\varphi_i$
	
	Since in Line 19 $M[\theta_k.parent,t]\leftarrow M[\theta_k.root,t]$ \\
	and $j=\theta_k.parent$ we have \\$M[j,t]\leftarrow M[i,t]$, as a result
	$$M[j,t]=M[\theta_k.parent,t]=\top\mbox{ iff }(\hat{\rho},t,\varepsilon)\models v_k.\varphi_i\equiv\varphi_j$$
\end{proof}

\subsection{Proof of the  correctness of Algorithm \ref{alg:DP-LTL}}
\label{proofLTL}
LTL formulas consider only propositional and temporal operators; therefore, the time variables' environment ($\varepsilon$) is not affected by Algorithm \ref{alg:DP-LTL}. 
Since time variables do not change during Algorithm \ref{alg:DP-LTL}, we assume that Algorithm \ref{alg:DP-LTL} considers time constraints as $\top/\bot$ values since they are already evaluated in Algorithm \ref{alg:off-line}.
In this section, we prove that the output of Algorithm \ref{alg:DP-LTL} corresponds to the correct evaluation of the LTL subformula  $\varphi_j$ at sample instance $u$ based on Definition \ref{def:sem4tptl}. 

In essence, we will prove $M[j,u]=\top$ if $(\hat{\rho},u,\varepsilon)\models\varphi_j$ and similarly $M[j,u]=\bot$ if $(\hat{\rho},u,\varepsilon)\not\models\varphi_j$.
For the proof of Algorithm \ref{alg:DP-LTL}, we use induction:

{\bf Base:} In Section \ref{MA}, we mentioned that in Line 1 of Algorithm \ref{alg:off-line} the corresponding values for atomic propositions are stored in the monitoring table. 
In essence, for each $a\in AP$, and for each time stamp $\tau_u$, we save the following values in the monitoring table entry $M[a_{index},u]$, where $a_{index}$ is the index of atomic proposition $a$ in the monitoring table $M_{|\varphi|\times |\hat{\rho}|}$:
\begin{enumerate}
	\item $M[a_{index},u]\leftarrow\top$ if $a\in \dsig_u$ if $(\hat{\rho},u,\varepsilon)\models a$
	\item $M[a_{index},u]\leftarrow\bot$ if $a\not\in \dsig_u$ if $(\hat{\rho},u,\varepsilon)\not\models a$
\end{enumerate}
Since evaluation of predicates is independent of the time variables' environment ($\varepsilon$) the above cases are always satisfied for all sample instances $u$ and all environments $\varepsilon$.
As a result, every table entry corresponding to a predicate, correctly reflects the satisfaction of the predicate with respect to the state trace $\hat{\dsig}$ and the environment $\varepsilon$.
Similarly, the table entries for constant Boolean values ($\top/\bot$) are trivially correct.

{\bf Induction Hypothesis:} Algorithm \ref{alg:off-line} updates the values of Table from right to left, i.e., for the samples with indexes $|\hat{\rho}|-1$ down to 0.
This is because we resolve temporal operators looking into the future. 
Namely, if the Boolean value in the next samples of time are resolved, then we can resolve the Boolean evaluation for the current sample of time.
For the Induction Hypothesis, we assume the table entries for the proper subformulas of $\varphi_j$ at the same or future samples contain the correct $\top/\bot$, i.e, we assume that\\
$$\forall \varphi_k\subset \varphi_j, \forall v\ge u, M[k,v]=\top\mbox{ iff }(\hat{\rho},v,\varepsilon)\models \varphi_k$$
And also for the same subformula ($\varphi_j$), we assume the table entries for all the future samples contain the correct $\top/\bot$ values as follows:\\
$$\forall v>u, M[j,v]=\top\mbox{ iff }(\hat{\rho},v,\varepsilon)\models \varphi_j$$

{\bf Induction Step:} For the induction step we consider five cases of $\varphi_j$:\\
\\{\bf Case 1:} $\varphi_j\equiv\neg\varphi_m$: \\
Consider $M[j,u]\leftarrow\neg M[m,u]$ (Algorithm \ref{alg:DP-LTL}, Line 2).\\ 
According to Definition \ref{def:sem4tptl}:
$(\hat{\rho},u,\varepsilon)\models\neg\varphi_m$  iff  $(\hat{\rho},u,\varepsilon)\not\models\varphi_m$\\ 
Based on IH: $M[m,u]=\bot$ iff $(\hat{\rho},u,\varepsilon)\not\models\varphi_m$  iff (based on Def. \ref{def:sem4tptl}) $(\hat{\rho},u,\varepsilon)\models\neg\varphi_m\equiv\varphi_j$ \\ 
Therefore, $M[j,u]=\neg M[m,u]=\neg\bot$ iff $(\hat{\rho},u,\varepsilon)\not\models\varphi_m$ iff  $(\hat{\rho},u,\varepsilon)\models\neg\varphi_m\equiv\varphi_j$\\
As a result $M[j,u]=\top$ iff $(\hat{\rho},u,\varepsilon)\models\varphi_j$\\
\\{\bf Case 2:}
$\varphi_j\equiv\varphi_m\wedge\varphi_n$:\\
Consider $M[j,u]\leftarrow M[m,u]\wedge M[n,u]$ (Algorithm \ref{alg:DP-LTL}, Line 4).\\
According to Definition \ref{def:sem4tptl}: $(\hat{\rho},u,\varepsilon)\models\varphi_m\wedge\varphi_n$ iff  $(\hat{\rho},u,\varepsilon)\models\varphi_m$ and $(\hat{\rho},u,\varepsilon)\models\varphi_n$\\
Based on IH: $M[m,u]=\top$ iff $(\hat{\rho},u,\varepsilon)\models\varphi_m$ and $M[n,u]=\top$ iff $(\hat{\rho},u,\varepsilon)\models\varphi_n$\\
We know that, $M[m,u]\wedge M[n,u]=\top$ iff $M[m,u]=\top$ and $M[n,u]=\top$\\
Thus, $M[m,u]\wedge M[n,u]=\top$ iff $(\hat{\rho},u,\varepsilon)\models\varphi_m$ and $(\hat{\rho},u,\varepsilon)\models\varphi_n$\\
Therefore, $M[m,u]\wedge M[n,u]=\top$ iff $(\hat{\rho},u,\varepsilon)\models\varphi_m\wedge\varphi_n\equiv\varphi_j$\\
As a result $M[j,u]=\top$ iff $(\hat{\rho},u,\varepsilon)\models\varphi_j$\\
\\{\bf Case 3:}
$\varphi_j\equiv\varphi_m\vee\varphi_n$:\\
Consider $M[j,u]\leftarrow M[m,u]\vee M[n,u]$ (Algorithm \ref{alg:DP-LTL}, Line 6).\\
According to Definition \ref{def:sem4tptl}: $(\hat{\rho},u,\varepsilon)\models\varphi_m\vee\varphi_n$  iff  $(\hat{\rho},u,\varepsilon)\models\varphi_m$ or $(\hat{\rho},u,\varepsilon)\models\varphi_n$\\
Based on IH: $M[m,u]=\top$ iff $(\hat{\rho},u,\varepsilon)\models\varphi_m$ and $M[n,u]=\top$ iff $(\hat{\rho},u,\varepsilon)\models\varphi_n$\\
We know that, $M[m,u]\vee M[n,u]=\top$ iff $M[m,u]=\top$ or $M[n,u]=\top$\\
Thus, $M[m,u]\vee M[n,u]=\top$ iff $(\hat{\rho},u,\varepsilon)\models\varphi_m$ or $(\hat{\rho},u,\varepsilon)\models\varphi_n$\\
Therefore, $M[m,u]\vee M[n,u]=\top$ iff $(\hat{\rho},u,\varepsilon)\models\varphi_m\vee\varphi_n\equiv\varphi_j$\\
As a result $M[j,u]=\top$ iff $(\hat{\rho},u,\varepsilon)\models\varphi_j$\\
\\{\bf Case 4:}
$\varphi_j\equiv\bigcirc\varphi_m$\\
Consider $M[j,u]\leftarrow M[m,u+1]$ if $u<|\hat{\rho}|-1$ (Line 11) and $M[j,u]\leftarrow\bot$ otherwise (Line 9 of Algorithm \ref{alg:DP-LTL}).\\
According to Definition \ref{def:sem4tptl} we have two cases:\\
{\bf Case 4.1)} $u<(|\hat{\rho}|-1)$:\\ $(\hat{\rho},u,\varepsilon)\models\bigcirc\varphi_m$ iff $(\hat{\rho},u+1,\varepsilon)\models\varphi_m$\\
Based on IH: $M[m,u+1]=\top$ iff $(\hat{\rho},u+1,\varepsilon)\models\varphi_m$ iff $(\hat{\rho},u,\varepsilon)\models\bigcirc\varphi_m\equiv\varphi_j$\\ 
As a result $M[j,u]=M[m,u+1]=\top$ iff $(\hat{\rho},u,\varepsilon)\models\varphi_j$\\
{\bf Case 4.2)} $u=|\hat{\rho}|-1$:\\ 
by Definition \ref{def:sem4tptl}, $(\hat{\rho},u,\varepsilon)\not\models\bot$\\
Line 9 of Algorithm \ref{alg:DP-LTL} similarly assigns $M[j,u]\leftarrow\bot$\\
\\{\bf Case 5:}
$\varphi_j\equiv\varphi_mU\varphi_n$\\
According to \cite{Gerth1995SOA}, Until operation can be simplified according to following equivalence relation: 
$$\phi U\psi\equiv\psi\vee(\phi\wedge\bigcirc(\phi U\psi))$$
In other words, we need to consider current value of $\bigcirc(\phi U\psi)$ (future value of $\phi U\psi$ at the next sample) and use the current values of $\phi$ and $\psi$ to resolve and evaluate $\phi U\psi$ at the current sample using  equation $\psi\vee(\phi\wedge\bigcirc(\phi U\psi))$. Algorithm \ref{alg:DP-LTL} considers two case for $\varphi_j\equiv\varphi_mU\varphi_n\equiv\varphi_n\vee(\varphi_m\wedge\bigcirc(\varphi_m U\varphi_n))$:\\
{\bf Case 5.1)} $u<(|\hat{\rho}|-1)$:\\
Now consider the update of $M[j,u]\leftarrow M[n,u]\vee(M[m,u]\wedge M[j,u+1])$ according to Line 17 of Algorithm \ref{alg:DP-LTL}.\\
Based on IH: $M[n,u]=\top$ iff $(\hat{\rho},u,\varepsilon)\models\varphi_n$ and $M[m,u]=\top$ iff $(\hat{\rho},u,\varepsilon)\models\varphi_m$ and\\ $M[j,u+1]=\top$ iff $(\hat{\rho},u+1,\varepsilon)\models\varphi_j$ iff $(\hat{\rho},u,\varepsilon)\models\bigcirc\varphi_j$\\
According to Case 2 (Conjunction)  $M[m,u]\wedge M[j,u+1]=\top$ iff $(\hat{\rho},u,\varepsilon)\models\varphi_m$ and $(\hat{\rho},u,\varepsilon)\models\bigcirc\varphi_j$\\
Therefore, $M[m,u]\wedge M[j,u+1]=\top$ iff $(\hat{\rho},u,\varepsilon)\models\varphi_m\wedge\bigcirc\varphi_j$ \\
We know that, $M[j,u]=\top$ iff $M[n,u]=\top$ or $M[m,u]\wedge M[j,u+1]=\top$ \\
According to Case 3 (Disjunction) $M[j,u]=\top$ iff $(\hat{\rho},u,\varepsilon)\models\varphi_n$ or $(\hat{\rho},u,\varepsilon)\models\varphi_m\wedge\bigcirc\varphi_j$ \\
As a result, $M[j,u]=\top$ iff $(\hat{\rho},u,\varepsilon)\models\varphi_n\vee(\varphi_m\wedge\bigcirc\varphi_j)$\\
{\bf Case 5.2)} $u=|\hat{\rho}|-1$:\\
According to Case 4.2 for Next operator: $(\hat{\rho},u,\varepsilon)\not\models\bot$\\
This implies that $\varphi_j\equiv\varphi_n\vee(\varphi_m\wedge\bot)\equiv\varphi_n\vee\bot\equiv\varphi_n$\\
Now consider the update of $M[j,u]\leftarrow M[n,u]$ according to Line 15 of Algorithm \ref{alg:DP-LTL}.\\
Based on IH: $M[n,u]=\top$ iff $(\hat{\rho},u,\varepsilon)\models\varphi_n$\\
Therefore after the assignment, $M[j,u]=\top$ iff $(\hat{\rho},u,\varepsilon)\models\varphi_j$

\end{document}